\pdfoutput=1
\pdfpagewidth=\paperwidth
\pdfpageheight=\paperheight
\documentclass[10pt, letterpaper]{article}

\DeclareFixedFont{\MyTitleFont}{OT1}{ptm}{m}{n}{20pt}
\DeclareFixedFont{\MyAuthorFont}{OT1}{ptm}{m}{n}{12pt}
\DeclareFixedFont{\MyAbstractTitleFont}{OT1}{ptm}{m}{n}{11pt}
\DeclareFixedFont{\MyAbstractFont}{OT1}{ptm}{m}{n}{10pt}
\DeclareFixedFont{\MySubtitleFont}{OT1}{ptm}{m}{n}{14pt}
\DeclareFixedFont{\MySubSubtitleFont}{OT1}{ptm}{m}{n}{12pt}
\DeclareFixedFont{\MySubSubSubtitleFont}{OT1}{ptm}{m}{n}{10pt}
\DeclareFixedFont{\MyTextFont}{OT1}{ptm}{m}{n}{10pt}

\usepackage{geometry}
\emergencystretch = \maxdimen
\usepackage{titling}

\title{\MyTitleFont Partially Independent Control Scheme for Spacecraft Rendezvous in Near-Circular Orbits \vspace{0em}}

\author{\MyAuthorFont Neng Wan and Weiran Yao}
\date{}

\usepackage{titlesec}

\usepackage[marginal]{footmisc}
\usepackage{abstract}

\usepackage[sort]{cite}

\usepackage{indentfirst}
\usepackage[round]{natbib}
\usepackage{amsthm}
\usepackage{amsmath}
\usepackage{amsfonts}
\usepackage{txfonts}
\usepackage{bm}
\usepackage{mathptmx}
\usepackage{enumerate}
\usepackage{lineno}
\usepackage[pdftex, bookmarks=false]{hyperref}
\hypersetup{colorlinks = true, citecolor = blue, linkcolor = blue, urlcolor = black}

\usepackage[pdftex]{color,graphicx}

\theoremstyle{remark}
\newtheorem*{notation}{Notation}
\newtheorem{remark}{Remark}
\theoremstyle{definition}
\newtheorem{lemma}{Lemma}
\newtheorem{assumption}{Assumption}
\newtheorem{theorem}{Theorem}

\begin{document}

\maketitle

\footnotetext[0]{Neng Wan is with Department of Mathematics and Statistics, University of Minnesota Duluth, Duluth 55811, USA (corresponding author). Email: \href{wanxx179@d.umn.edu}{wanxx179@d.umn.edu}.}
\footnotetext[0]{Weiran Yao is with School of Astronautics, Harbin Institute of Technology, Harbin 150001, China. Email: \href{yaowr1990@163.com}
{yaowr1990@163.com}.}

\vspace{-4em}

\begin{abstract}
Due to the complexity and inconstancy of the space environment, accurate mathematical models for spacecraft rendezvous are difficult to obtain, which consequently complicates the control tasks. In this paper, a linearized time-variant plant model with external perturbations is adopted to approximate the real circumstance. To realize the robust stability with optimal performance cost, a partially independent control scheme is proposed, which consists of a robust anti-windup controller for the in-plane motion and a ${{H}_{\infty}}$ controller for the out-of-plane motion. Finally, a rendezvous simulation is given to corroborate the practicality and advantages of the partially independent control scheme over a coupled control scheme.

\vspace{0.5em}

\hspace{-1.6em}{\it Keywords}: Spacecraft rendezvous; Near-circular orbits; Partially independent control; Robust control; Anti-windup.
\end{abstract}

\titleformat*{\section}{\centering\MySubtitleFont}
\titlespacing*{\section}{0em}{1.25em}{1.25em}[0em]
\section{Introduction}
Widely applied to crew exchange, large-scale assembly, spacecraft maintenance, docking, interception, formation flying and other astronautic missions involving more than one spacecraft, autonomous spacecraft rendezvous has been regarded as a crucial operational technology in aerospace engineering. As the autonomous control scheme is a cardinal and decisive issue that determines the success of the rendezvous, it has been and continues to be an engaging area of study.

Most of the mathematical models employed in investigating spacecraft rendezvous are derived from the two-body problem. Because of their concise and linearized form, Clohessy-Wiltshire equations \citetext{\citealp{W.H.Clohessy_JAS_1960}} were favored by many researchers, though this model was initially developed to describe the rendezvous in circular orbits. The models put forward by \cite{J.P.DeVries_AIAAJ_1963} and \cite{J.Tschauner_AIAAJ_1967} extended our knowledge to the rendezvous in elliptical orbits; however, nonlinear terms were involved, which circumscribed their broader implementations in control engineering. Considering the fact that most of the rendezvous missions were conducted in near-circular orbits with small eccentricities, researchers began to search for some eclectic models that are linearized and sufficiently precise. A comprehensive survey on these efforts was given by \cite{T.E.Carter_JGCD_1998}; nevertheless, all the linearization results introduced in this literature are in terms of either the true or eccentric anomaly of one spacecraft and require the solution of the Kepler problem, which is time and computational consuming. A time-explicit dynamical model overcoming this defect was first introduced by \cite{Anthony_AIAAJ_1965}, and a more recent development on time-explicit models was contributed by \cite{R.G.Melton_JGCD_2000}.

Robust guaranteed cost control was first raised by \cite{S.S.L.Chang_TAC_1972} to optimize preassigned cost function, and many of the following literatures were carried out based on their works. \cite{I.R.Petersen_TAC_1994} synthesized a state feedback guaranteed cost controller via a Riccati equation approach. \cite{L.Yu_Auto_1999} designed a guaranteed cost controller for linear uncertain time-delay systems via a linear matrix inequality (LMI) method. \cite{S.H.Esfahani_IJRNC_2000} solved the guaranteed cost output feedback control problem in a matrix substitution manner. More recently, \cite{X.P.Guan_TFS_2004} and \cite{L.Wu_OCAM_2011} investigated the guaranteed cost control methods for time-delay systems. \cite{H.Zhang_TSMC_2008} studied a guaranteed cost control scheme for a class of uncertain stochastic nonlinear systems with multiple time delays. \cite{K.Tanaka_TSMC_2009} presented a guaranteed cost control for polynomial fuzzy systems via a sum of squares approach.

Robust ${{H}_{\infty}}$ control technique is frequently used in synthesizing guaranteed cost controllers for systems with external disturbances. This technique was first proposed by \cite{G.Zames_TAC_1981}. Nonetheless, the focus of the ${{H}_{\infty}}$ control problem quickly shifted from its applications to formulating the solvable control problems due to the lack of an efficient tool to solve the ${{H}_{\infty}}$ problem. Time domain approach \citetext{\citealp{B.R.Barmish_TAC_1983}}, frequency domain approach \citetext{\citealp{B.A.Francis_Springer_1987}} and Riccati equation approach \citetext{\citealp{P.P.Khargonekar_TAC_1990}} were the three main methods in solving the $H_\infty$ control problems before the LMI approach \citetext{\citealp{S.P.Boyd_SIAM_1994, M.Chilali_TAC_1996}} became widely used. For more recent papers on robust ${{H}_{\infty}}$ control, refer to \cite{M.Liu_SP_2011}, \cite{L.Wu_JFI_2011} and references therein.

Optimal spacecraft rendezvous problem has attracted numerous researchers. Some previous works on this topic have been introduced in \cite{N.Wan_MPE_2013}. Based on the sliding mode control theory, \cite{B.Ebrahimi_AA_2008} and \cite{L.Zhao_CDC_2013} developed the optimal guidance laws for spacecraft rendezvous. \cite{H.Gao_TCST_2009} investigated a multi-object robust ${{H}_{\infty}}$ control scheme for rendezvous in circular orbits. \cite{Z.Li_MPE_2013a} proposed a sample-data control technique for rendezvous via a discontinuous Lyapunov approach. \cite{X.Yang_AST_2013} synthesized a robust reliable controller for thrust-limited rendezvous in circular orbits. \cite{X.Gao_JFI_2012} studied a robust ${{H}_{\infty}}$ control approach for rendezvous in elliptical orbits. \cite{X.Yang_TIE_2012} considered the spacecraft rendezvous with thrust nonlinearity and sampled-data control. \cite{N.Wan_MPE_2013} put forward a robust tracking control method for relative position holding and rendezvous with actuator saturation in near-circular orbits; and in another paper of \cite{N.Wan_AAA_2014}, they provided an observer-based control scheme for spacecraft rendezvous. More recent works on optimal spacecraft rendezvous can be found in \cite{Z.Li_MPE_2013b}, \cite{Z.Li_JFI_2013c}, \cite{D.Sheng_MPE_2014} and \cite{Zhou_TCST_2014}. Nevertheless, to the best of the authors' knowledge, most of the existing literatures either synthesized a coupled rendezvous controller that regulated the in-plane and out-of-plane motions jointly or neglected the control task of out-of-plane motion. Although the in-plane and out-of-plane motions were treated separately by \cite{X.Gao_IMAJMCI_2011}, an identical control method was applied to two motions. Therefore, up to now, an efficient control scheme which accommodates the different dynamical and engineering features of the in-plane and the out-of-plane motions has not been proposed yet.

In this paper, a time-explicit linearized model for rendezvous in near-circular orbits is established in a concise form that facilitates the controller synthesis; non-circularity of the reference orbits and external perturbations are considered to ensure the accuracy of the plant model. In-plane and out-of-plane motion controllers are synthesized respectively in order to meet the dynamical properties and requirements of each motion. For the in-plane motion usually driven by high-thrust propellers with high fuel consumption, a robust anti-windup guaranteed cost controller is synthesized to realize optimal rendezvous under the constraints of orbital non-circularity and actuator saturation. Moreover, it is well known that the out-of-plane maneuver or maneuver that changes the orbital inclination consumes much more energy compared with other kinds of orbital maneuvers \citetext{\citealp{H.Curtis_BH_2005}}; therefore a robust ${{H}_{\infty}}$ controller is synthesized to guarantee the robust stability of the out-of-plane motion, which is usually driven by low-thrust propellers thus very sensitive to the external disturbances. Then the partially independent controller is obtained by solving two convex optimization problems subject to LMI constraints. At the end of this paper, a numerical rendezvous simulation is presented to verify the advantages of the partially independent control scheme over a coupled robust controller.

The remainder of this paper is organized as follows. \hyperref[sec2]{Section 2} establishes the dynamical models and formulates the control problems; \hyperref[sec3]{Section 3} shows the main result of the partially independent control scheme; \hyperref[sec4]{Section 4} presents a numerical simulation; and \hyperref[sec5]{Section 5} draws the conclusion.

\begin{notation}
The notations used throughout this paper are defined in this paragraph. $\|\cdot\|_2$ refers to the Euclidean vector norm. diag($\cdots$) stands for a block-diagonal matrix. In symmetric block matrices or complex matrix expressions, an asterisk ($*$) is used to represent a term that is induced by symmetry. For a matrix $\boldsymbol{A}$, $\boldsymbol{A}^T$ stands for the transpose of $\boldsymbol{A}$; and sym$(\boldsymbol{A})$ stands for $\boldsymbol{A} + \boldsymbol{A}^T$ when $\boldsymbol{A}$ is a square matrix. For a real symmetric matrix $\boldsymbol{B}$, the notation $\boldsymbol{B} > \mathbf{0}$ ($\boldsymbol{B} < \mathbf{0}$) is used to denote its positive- (negative-) definiteness. $\boldsymbol{I}$ and $\mathbf{0}$ respectively denote the identity matrix and zero matrix with compatible dimension. If the dimensions of matrices are not explicitly stated, they are assumed to be compatible for algebraic operation.
\end{notation}

\section{Dynamical Model and Problem Formulation}\label{sec2}
In this section, dynamical models for the in-plane and out-of-plane motions are established, and the control problems are formulated with the consideration of the different dynamical features and engineering demands of each motion.

Suppose that a target vehicle is moving on a near-circular orbit with a chase vehicle nearby. Both of the spacecrafts are only influenced by a central gravitational source, and the target vehicle does not maneuver during the rendezvous. A relative Cartesian coordinate system adopted to describe the relative motion between the spacecrafts is defined in \autoref{fig1}. The system's origin is fixed at the centroid of the target vehicle. The $x$-axis is parallel to the vector $\boldsymbol{r}$ from the Earth's centroid to the target's centroid; $\boldsymbol{r}_c$ is the vector from the Earth's centroid to the chaser's centroid. The $z$-axis is aligned with the target orbit's angular momentum vector, and the $y$-axis completes a right-handed coordinate system. 

\begin{figure}[!h]
\centering
\includegraphics[width=0.5\textwidth]{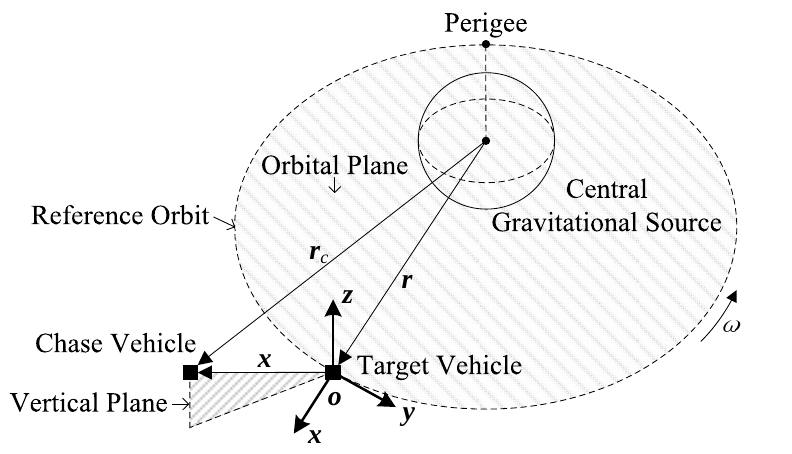}
\caption{Relative Cartesian coordinate system for spacecraft rendezvous.}\label{fig1}
\end{figure}

\noindent Some other important assumptions employed in this paper are also presumed here in case of ambiguity. 
\begin{assumption} \label{asm1}
The propulsions of the chase vehicle are continuous and independent along each axis defined in \hyperref[fig1]{Figure 1}.
\end{assumption}
\begin{assumption} \label{asm2}
The initial out-of-plane distance and velocity between the chase and target spacecrafts are zeros. 
\end{assumption}
\begin{assumption} \label{asm3}
Only the disturbance along the $z$-axis is included in the plant model, {\it i.e.}, external perturbations along the orbital plane are neglected in this paper.
\end{assumption}
\begin{remark}
As redundancy is a fundamental technology for spacecraft system, independent propulsions can be realized with a proper actuator allocation; therefore {\hyperref[asm1]{Assumption~1}} is frequently taken in the existing literatures, such as \cite{B.Ebrahimi_AA_2008}, \cite{H.Gao_TCST_2009} and \cite{Zhou_TCST_2014}. Since out-of-plane maneuver that changes orbital inclination is fuel consuming (for example, when both the initial and terminal orbits are circular, the velocity increase $\upDelta v$ required for an inclination change $\upDelta i$ is $\upDelta v = 2v\sin(\upDelta i/2)$, where $v$ is the orbital velocity in a large magnitude.), the relative distance and velocity along the $z-$axis are often eliminated by launch vehicle before the close-range rendezvous, the phase we mainly investigates in this paper; therefore, {\hyperref[asm2]{Assumption~2}} is reasonable. As it was mentioned above, for the out-of-plane motion, due to its dynamical and engineering properties, the fuel consumption and stability are more sensitive to the external perturbations compared with the other motions along the orbital plane; therefore, it is reasonable for us to conduct a special investigation in the crucial one while omitting the trivial ones, which is the main propose of {\hyperref[asm3]{Assumption~3}}.
\end{remark}

\subsection{Relative Motion Model}\label{sec2_1}
Define the state vector as $\bm{x}(t) = [x,\ y,\ z,\ \dot{x},\ \dot{y},\ \dot{z}]^T$, which contains the relative distances and velocities along each axis; and define the control input vector as $\bm{u}(t) = [f_x,\ f_y,\ f_z]^T$, where $f_i$ for $i = x, y, z$ are the control forces acting on the chase vehicle along each axis. Relative motion models for spacecraft rendezvous in all types of conic orbits can be uniformly expressed in a matrix form as
\begin{equation} \label{sec2_eq1}
\dot{\bm{x}}(t) = \bm{A}_n\bm{x}(t)+\bm{B}\bm{u}(t) \  \textrm{,}
\end{equation}
\noindent where
\begin{flalign*}
\begin{split}
\indent
\bm{A}_n =
\left[\begin{matrix}
0 & 0 & 0 & 1 & 0 & 0\\
0 & 0 & 0 & 0 & 1 & 0\\
0 & 0 & 0 & 0 & 0 & 1\\
2\mu /{{r}^{3}}+{{\omega }^{2}} & {\dot{\omega }} & 0 & 0 & 2\omega  & 0\\
-\dot{\omega } & -\mu /{{r}^{3}}+{{\omega }^{2}} & 0 & -2\omega  & 0 & 0\\
0 & 0 & -\mu /{{r}^{3}} & 0 & 0 & 0\\
\end{matrix}\right]
\textrm{,}
\qquad
\bm{B} = 
\frac{\displaystyle 1}{\displaystyle m}
\left[\begin{matrix}
0 & 0 & 0\\
0 & 0 & 0\\
0 & 0 & 0\\
1 & 0 & 0\\
0 & 1 & 0\\
0 & 0 & 1\\
\end{matrix}\right]
\textrm{.}
\end{split}&
\end{flalign*}

\noindent $\mu$ is the gravitational parameter; $r$ is the radius of the reference orbit; $\omega$ and $\dot{\omega}$ are the angular rate and angular acceleration of the target vehicle; and $m$ is the mass of the chase vehicle. As can be seen in~\eqref{sec2_eq1}, nonlinear terms exist in system matrix $\bm{A}_n$, which makes the controller synthesis difficult. Therefore, a further linearization on ~\eqref{sec2_eq1} is necessary, and a lemma known as generalized Lagrange's expansion theorem is introduced here before the linearization procedures.

\begin{lemma}[\citealp{Battin_1999}]\label{Lemma_Lagrange}
Let $y$ be a function of $x$ in terms of a parameter $\alpha$ by
\begin{equation} \label{sec2_eq2}
y = x + \alpha\phi(y) \ \textrm{.}
\end{equation}
Then for sufficiently small $\alpha$, any function $F(y)$ can be expanded as a power series in $\alpha$,
\begin{equation} \label{sec2_eq3}
F(y) = F(x) + \sum_{n=1}^{\infty}\frac{\alpha^n}{n!}\frac{{\rm d}^{n-1}}{{\rm d}x^{n-1}}\left[\phi(x)^n\frac{{\rm d}F(x)}{{\rm d}x} \right]
\textrm{.}
\end{equation}
\end{lemma}

With equation
\begin{equation} \label{sec2_eq4}
r = a(1-e\cos E) \ \textrm{,}
\end{equation}
where $a$ and $e$ denote the semimajor axis and the eccentricity of the reference orbit; $E$ denote the eccentric anomaly of the target vehicle; nonlinear terms in system matrix $\bm{A}_n$ can be rewritten as the functions in terms of $E$, such as

\nolinenumbers
\begin{subequations}\label{sec2_eq5}
\begin{equation}
\frac{\mu}{r^3} = n^2\left(\frac{a}{r}\right)^3 = n^2\left(\frac{1}{1-e\cos E}\right)^3\textrm{,}
\end{equation}
\begin{equation}
\omega = \frac{h}{r^2} = n\left(\frac{1}{1-e\cos E}\right)^2\textrm{,}
\end{equation}
\begin{equation}
\omega^2 = \frac{h^2}{r^4} = n^2\left(\frac{1}{1-e\cos E}\right)^4\textrm{,}
\end{equation}
\begin{equation}
\dot{\omega} = -\frac{2h}{r^3}\dot{r} = -2n^2\frac{e\sin E}{\left(1-e\cos E\right)^4}\textrm{,}
\end{equation}
\end{subequations}
\noindent where $h$ is the angular momentum of the reference orbit. Moreover, according to \hyperref[Lemma_Lagrange]{Lemma~1} and Kepler's time equation
\begin{equation} \label{sec2_eq6}
E = M + e\sin E \ \textrm{,}
\end{equation}
\noindent where $M = n(t-t_p)$ and $n = \sqrt{\mu / a^3}$ are the mean anomaly and the mean motion of the target vehicle respectively; $t_p$ is the time of periapsis passage; when eccentricity $e$ is sufficiently small, any function $F(E)$ can be expanded as a power series in constant $e$. Therefore, equations (\hyperref[sec2_eq5]{5a-d}) can be expanded as

\nolinenumbers
\begin{subequations}\label{sec2_eq7}
\begin{equation}
\frac{\mu}{r^3}=n^2 \left\{ \frac{1}{\left(1-e\cos M\right)^3} - \frac{3e^2\sin^2 M}{\left(1-e\cos M\right)^4} + \frac{e^2}{2}\left[\frac{12e^2\sin^4M}{\left(1-e\cos M\right)^5} - \frac{9e\cos M\sin^2 M}{\left(1-e\cos M\right)^4} \right] + \cdots \right \}\textrm{,}
\end{equation}
\begin{equation}
\omega = n \left\{\frac{1}{\left(1-e\cos M \right)^2} - \frac{2e^2\sin^2 M}{\left(1-e\cos M \right)^3} + \frac{e^2}{2}\left[\frac{6e^2\sin^4 M}{\left(1-e\cos M\right)^4}  - \frac{6e\cos M\sin^2 M}{\left(1-e\cos M\right)^3} \right] + \cdots \right\}\textrm{,}
\end{equation}
\begin{equation}
\omega^2 = n^2\left\{\frac{1}{\left(1-e\cos M\right)^4} - \frac{4e^2\sin^2 M}{\left(1-e\cos M\right)^5} + \frac{e^2}{2}\left[\frac{20e^2\sin^4 M}{\left(1-e\cos M\right)^6} - \frac{12e\cos M\sin^2 M}{\left(1-e\cos M\right)^5}  \right] + \cdots \right\}\textrm{,}
\end{equation}
\begin{equation}
\dot{\omega} = -2n^2\left\{\frac{e\sin M}{\left(1-e\cos M\right)^4} + e^2 \sin M \left[\frac{\cos M}{\left(1-e\cos M\right)^4} - \frac{4e\sin^2 M}{\left(1-e\cos M\right)^5}\right] + \cdots \right\}\textrm{.}
\end{equation}
\end{subequations}
\noindent Computing the Taylor series expansions of~(\hyperref[sec2_eq7]{7a-d}) around point $e = 0$, we have

\nolinenumbers
\begin{subequations}\label{sec2_eq8}
\begin{equation}
\frac{\mu}{r^3} = n^2\left[1+3e\cos M + \frac{e^2}{2}\left(9\cos 2M + 3\right) + \frac{e^3}{8}\left(53\cos 3M + 27\cos M\right) + O\left(e^4\right)   \right]\textrm{,}
\end{equation}
\begin{equation}
\omega = n\left[1+2e\cos M + \frac{e^2}{2}\left(5\cos2M + 1\right) + \frac{e^3}{4} \left(13\cos3M+3\cos M\right) + O\left(e^4\right)  \right]\textrm{,}
\end{equation}
\begin{equation}
\omega^2 = n^2\left[1+4e\cos M + e^2\left(7\cos2M+3\right) + \frac{e^3}{2}\left(23\cos3M+17\cos M\right) + O\left(e^4\right)  \right]\textrm{,}
\end{equation}
\begin{equation}
\dot{\omega} = -2n^2\left[e\sin M + \frac{5e^2}{2}\sin 2M + e^3 \left(\frac{23}{8}\sin 3M + 4\cos 2M \sin M + \frac{19}{8}\sin M  \right) + O\left(e^4\right) \right]\textrm{.}
\end{equation}
\end{subequations}
\noindent Truncating the expansions~(\hyperref[sec2_eq8]{8a-d}) to order $e$ and substituting the results into \eqref{sec2_eq1}, the linearized relative motion model becomes
\begin{equation}\label{sec2_eq9}
\dot{\bm{x}}(t) = (\bm{A} + \upDelta\bm{A})\bm{x}(t)+\bm{B}\bm{u}(t) \ \textrm{,}
\end{equation}
\noindent where
\begin{flalign*}
\begin{split}
\indent&
\bm{A} =
\left[\begin{matrix}
0 & 0 & 0 & 1 & 0 & 0\\
0 & 0 & 0 & 0 & 1 & 0\\
0 & 0 & 0 & 0 & 0 & 1\\
3n^2 &0 & 0 & 0 & 2n & 0\\
0 & 0 & 0 & -2n & 0 & 0\\
0 & 0 & -n^2 & 0 & 0 & 0\\
\end{matrix}\right]
\textrm{,}
\qquad
\bm{B} = 
\frac{\displaystyle 1}{\displaystyle m}
\left[\begin{matrix}
0 & 0 & 0\\
0 & 0 & 0\\
0 & 0 & 0\\
1 & 0 & 0\\
0 & 1 & 0\\
0 & 0 & 1\\
\end{matrix}\right]
\textrm{,}\\
\indent&
\upDelta\bm{A} = 
\left[\begin{matrix}
0 & 0 & 0 & 0 & 0 & 0\\
0 & 0 & 0 & 0 & 0 & 0\\
0 & 0 & 0 & 0 & 0 & 0\\
10en^2\cos M & -2en^2\sin M & 0 & 0 & 4en\cos M & 0\\
2en^2\sin M & en^2\cos M & 0 & -4en\cos M & 0 & 0\\
0 & 0 & -3en^2\cos M & 0 & 0 & 0\\
\end{matrix}\right]
\textrm{.}
\end{split}&
\end{flalign*}
\noindent The time-variant and norm-bounded matrix $\upDelta\bm{A}$ is defined as the non-circularity matrix, which contains the shape information of the reference orbit. Hereby, we have finished the linearization procedures.

\begin{remark}
Compared with C-W equations \citetext{\citealp{W.H.Clohessy_JAS_1960}}, the non-circularity matrix $\upDelta\bm{A}$ makes the model~\eqref{sec2_eq9} more accurate and practical for engineering applications, while compared with the nonlinear model~\eqref{sec2_eq1}, the linearized representation of ~\eqref{sec2_eq9} makes the controller synthesis easier.
\end{remark}

\begin{remark}
Although the non-circularity matrix~$\upDelta\bm{A}$~in model~\eqref{sec2_eq9} is exactly known, it satisfies the matched condition usually employed to describe the unknown uncertainty, $\upDelta\bm{A} = \bm{DF}(t)\bm{E}$ \citetext{\citealp{P.P.Khargonekar_TAC_1990}}, which brings some convenience to controller synthesis. Therefore, this kind of time-variant matrices are sometimes treated as unknown uncertainty in some literatures, such as \cite{X.Yang_AST_2013}, \cite{Q.Wang_CCC_2014} and \cite{N.Wan_AAA_2014}.
\end{remark}

In order to construct a partially independent control scheme with the in-plane and out-of-plane controllers synthesized separately, we will decompose model~\eqref{sec2_eq9} and formulate the control problems with regard to each plane respectively in the rest of this section.

\subsubsection*{In-Plane Motion Model}
The state vector of in-plane motion is defined as $\bm{p}(t) = [x, \ y, \ \dot{x}, \ \dot{y}]^T$, and the control input vector is denoted as $\bm{u}_p(t) = [f_x, \ f_y]^T$. Then according to \eqref{sec2_eq9}, the mathematical model of in-plane motion can be extracted as
\begin{equation}\label{sec2_eq10}
\dot{\bm{p}}(t) = (\bm{A}_p + \upDelta\bm{A}_p)\bm{p}(t)+\bm{B}_p\bm{u}_p(t) \ \textrm{,}
\end{equation}
\noindent where
\begin{flalign*}
\begin{split}
\indent&
\bm{A}_p =
\left[\begin{matrix}
0 & 0 & 1 & 0\\
0 & 0 & 0 & 1\\
3n^2 & 0 & 0 & 2n\\
0 & 0 & -2n & 0\\
\end{matrix}\right]
\textrm{,}
\qquad
\bm{B}_p = 
\frac{\displaystyle 1}{\displaystyle m}
\left[\begin{matrix}
0 & 0\\
0 & 0\\
1 & 0\\
0 & 1\\
\end{matrix}\right]
\textrm{,}\\
\indent&
\upDelta\bm{A}_p =
\left[\begin{matrix}
0 & 0 & 0 & 0\\
0 & 0 & 0 & 0\\
10en^2\cos M & -2en^2\sin M & 0 & 4en\cos M\\
2en^2\sin M & en^2\cos M & -4en\cos M & 0\\
\end{matrix}\right]
\textrm{.}
\end{split}&
\end{flalign*}

\noindent The norm-bounded matrix $\upDelta\bm{A}_p$ can be factorized as
\begin{equation}\label{sec2_eq11}
\upDelta\bm{A}_p = \bm{E}_{p1}\boldsymbol{\upLambda}_p\bm{E}_{p2} \ \textrm{,}
\end{equation}
\noindent where $\bm{E}_{p1}$, $\bm{E}_{p2}$ and $\boldsymbol{\upLambda}_p$ are matrices with proper dimensions and satisfy $\boldsymbol{\upLambda}_p^T\boldsymbol{\upLambda}_p < \bm{I}$.

\subsubsection*{Out-of-Plane Motion Model}
The state vector of out-of-plane motion is defined as $\bm{q}(t) = [z, \ \dot{z}]^T$, and the control input is denoted as $u_q(t) = f_z$. According to {\hyperref[asm3]{Assumption~3}}, external disturbance $w_q(t)$ should be involved into the out-of-plane motion model extracted from \eqref{sec2_eq9}. Then the model can be expressed as
\begin{equation}\label{sec2_eq12}
\dot{\bm{q}}(t) = \left(\bm{A}_q + \upDelta\bm{A}_q\right)\bm{q}(t) + \bm{B}_{q}\left[u_q(t) + w_q(t)\right] \ \textrm{,}
\end{equation}

\noindent where
\begin{flalign*}
\begin{split}
\indent\bm{A}_q =
\left[\begin{matrix}
0 & 1\\
-n^2 & 0\\
\end{matrix}\right]
\textrm{,}
\qquad
\upDelta\bm{A}_q = 
\left[\begin{matrix}
0 & 0\\
-3en^2\cos M & 0\\
\end{matrix}\right]
\textrm{,}
\qquad
\bm{B}_{q} = 
\frac{\displaystyle 1}{\displaystyle m}
\left[\begin{matrix}
0\\
1\\
\end{matrix}\right]
\textrm{.}
\end{split}&
\end{flalign*}

\noindent The norm-bounded matrix $\upDelta\bm{A}_q$ can be factorized as
\begin{equation}\label{sec2_eq13}
\upDelta\bm{A}_q = \bm{E}_{q1}\boldsymbol{\upLambda}_q\bm{E}_{q2} \ \textrm{,}
\end{equation}
\noindent where $\bm{E}_{q1}$, $\bm{E}_{q2}$ and $\boldsymbol{\upLambda}_q$ are matrices with proper dimensions and satisfy $\boldsymbol{\upLambda}_q^T\boldsymbol{\upLambda}_q < \bm{I}$.

\begin{remark}
From equations \eqref{sec2_eq10} and \eqref{sec2_eq12}, it can be seen that the motions along the orbital plane are coupled, which means a coupled controller should be employed, while the motion along the $z$-axis can be governed by an independent controller. That is the reason why the authors name this method partially independent control scheme.
\end{remark}

\subsection{Problem Formulation}\label{sec2.2}
Robust stability, bounded propulsions and optimal cost function are the three main objectives we will consider when designing the partially independent control scheme. With these requirements, the control problems of in-plane and out-of-plane motions will be formulated successively as follows.

\subsubsection*{Control Problem for In-Plane Motion}
In order to assess the fuel and time consumptions of in-plane motion within a performance index, the quadratic cost function of in-plane motion is defined as
\begin{equation}\label{sec2_eq14}
J_p = \int_{0}^{\infty}\left[\bm{p}^T(t)\bm{Q}_p\bm{p}(t)+\bm{u}_p^T\bm{R}_p\bm{u}_p(t)\right]{\rm d}t \ \textrm{,}
\end{equation}
\noindent where the positive symmetric matrix $\bm{R}_p \in \mathbb{R}^{2\times2}$ is related to the fuel consumption; and the positive symmetric matrix $\bm{Q}_p \in \mathbb{R}^{4\times4}$ is related to the state convergence rate and the smoothness of trajectory \citetext{\citealp{Yang_JAE_2011}}. With two auxiliary matrices, $\bm{U}_{px} = [1, \ 0]^T[1, \ 0]$ and $\bm{U}_{py} = [0, \ 1]^T[0, \ 1]$, thrust constraints along the $x$- and $y$-axis can be formulated as
\begin{equation}\label{sec2_eq15}
\left| f_i \right| = \left|\bm{U}_{pi}\bm{u}(t)\right| \leq u_{pi,\max} \ \textrm{,}\qquad\left(i=x, \ y\right)\textrm{,}
\end{equation}

\noindent where $u_{pi,\max}$ are the maximum control forces that can be generated by the propellers along $i$-axis. With the motion model \eqref{sec2_eq10} and the requirements presented at the preliminary of \hyperref[sec2.2]{Section 2.2}, the control task of in-plane motion can be described as: design an anti-windup robust guaranteed cost controller such that
\begin{enumerate}[(i)]
\setcounter{enumi}{0}
\item\label{req1} In-plane motion system~\eqref{sec2_eq10} is asymptotically stable at $\bm{p}(t) = \mathbf{0}$, {\it i.e.}, the chase vehicle can eventually rendezvous with the target vehicle;
\item\label{req2} Quadratic cost function~\eqref{sec2_eq14} is minimal, {\it i.e.}, an optimal compromise between the fuel consumption and the state convergence rate shall be reached;
\item\label{req3} Control forces along the $x$- and $y$-axis should satisfy the saturation constraints~\eqref{sec2_eq15}.
\end{enumerate}

\subsubsection*{Control Problem for Out-of-Plane Motion}
In order to evaluate the fuel and time consumptions of out-of-plane motion within a performance index, the quadratic cost function for out-of-plane motion is defined as
\begin{equation}\label{sec2_eq16}
J_q = \int_{0}^{\infty}\left[\bm{q}^T(t)\bm{Q}_q\bm{q}(t)+u_q^{T}R_{q}u_q(t)\right]{\rm d}t \ \textrm{,}
\end{equation}

\noindent where $\bm{Q}_q$ and $R_q$ are the state weighting matrix and control weighting scale, which have the same functions as matrices $\bm{Q}_p$ and $\bm{R}_p$ introduced in~\eqref{sec2_eq14}. When external perturbation $w_q(t)$ is considered in~\eqref{sec2_eq12}, to keep the chase vehicle from deviating from the orbital plane, the capability of actuator $u_{q,\max}$ must be greater than the largest perturbation force $w_{q,\max}$. Moreover, to attenuate or to cancel the perturbation, out-of-plane propulsion~$u_q(t)$ should follow~$w_q(t)$ exactly; therefore, additional consideration of actuator saturation along the~$z$-axis is unnecessary. With the motion model~\eqref{sec2_eq12} and the requirements illustrated above, the control task of out-of-plane motion can be summarized as: design a robust $H_\infty$ controller such that
\begin{enumerate}[(i)]
\setcounter{enumi}{3}
\item\label{req4} Out-of-plane motion system~\eqref{sec2_eq12} is robustly stable at $\bm{q}(t) = \mathbf{0}$, {\it i.e.}, the chase vehicle can be stabilized on the reference orbital plane in the presence of non-circularity~$\upDelta\bm{A}_q$ and external perturbation~$w_q(t)$;
\item\label{req5} Quadratic cost function~\eqref{sec2_eq16} is minimal, {\it i.e.}, an optimal compromise between the fuel consumption and the state convergence rate shall be realized subject to the external perturbation~$w_q(t)$.
\end{enumerate}

\section{Partially Independent Control Scheme}\label{sec3}
In this section, an anti-windup robust guaranteed cost controller and a robust $H_\infty$ controller will be synthesized successively to construct the partially independent control scheme for spacecraft rendezvous. Firstly, a lemma that will be employed in the subsequent derivation is introduced here.
\begin{lemma}[\citealp{P.P.Khargonekar_TAC_1990}] \label{Lemma_LMI}
Given matrices $\bm{Y} = \bm{Y}^T$, $\bm{D}$ and $\bm{E}$ of appropriate dimensions,
\begin{equation}\label{sec3_eq17}
\bm{Y} + \bm{D}\bm{F}\bm{E} + \bm{E}^T\bm{F}^T\bm{D}^T < \mathbf{0} \ \textrm{,}
\end{equation}
for all  $\bm{F}$ satisfying  $\bm{F}^T\bm{F} \leq \bm{I}$, if and only if there exists a scalar $\varepsilon > 0$ such that
\begin{equation} \label{sec3_eq18}
\bm{Y} + \varepsilon\bm{D}\bm{D}^T + \varepsilon^{-1}\bm{E}^T\bm{E} < \mathbf{0} \  \textrm{.}
\end{equation}
\end{lemma}

\subsection{In-Plane Motion Controller}\label{sec3.1}
Consider the following state feedback control law
\begin{equation}\label{sec3_eq19}
\bm{u}_{p}(t) = - \bm{K}_p\bm{p}(t) \ \textrm{,}
\end{equation}
\noindent where $\bm{K}_p \in\mathbb{R}^{2\times4}$ is the state feedback gain matrix of in-plane motion controller. Substituting equation~\eqref{sec3_eq19} into the plant model~\eqref{sec2_eq10}, the closed-loop model for in-plane motion is
\begin{equation}\label{sec3_eq20}
\dot{\bm{p}}(t) = \left(\bm{A}_p + \upDelta\bm{A}_p - \bm{B}_p\bm{K}_p\right)\bm{p}(t) \ \textrm{.}
\end{equation}

\noindent Sufficient condition for the existence of a thrust-limited robust guaranteed cost controller is described in \hyperref[theorem1]{Theorem~1}.

\begin{theorem}\label{theorem1}
Consider the closed-loop system~\eqref{sec3_eq20} with the state feedback control law in~\eqref{sec3_eq19}. For a given initial state vector $\bm{p}(0)$, if there exist a positive symmetric matrix $\bm{X}_p \in \mathbb{R}^{4\times4}$, a matrix $\bm{Y}_p \in \mathbb{R}^{2\times4}$, positive scalars $\varepsilon_p$ and $\rho$ satisfying

\begin{equation}\label{sec3_eq21}
\left[\begin{matrix}
\text{sym}\left(\bm{A}_p\bm{X}_p - \bm{B}_p\bm{Y}_p\right) + \varepsilon_p^{}\bm{E}_{p1}^{}\bm{E}_{p1}^T & \bm{X}_p^{}\bm{E}_{p2}^T & \bm{Y}_p^T & \bm{X}_p\\
* & -\varepsilon_p\bm{I} & \mathbf{0} & \mathbf{0}\\
* & * & -\bm{R}_p^{-1} & \mathbf{0}\\
* & * & * &  -\bm{Q}_p^{-1}\\
\end{matrix}\right] < \mathbf{0} \ \textrm{,}
\end{equation}

\begin{equation}\label{sec3_eq22}
\left[\begin{matrix}
-\rho^{-1} & \rho^{-1}\bm{p}^T(0)\\
* & -\bm{X}_p\\
\end{matrix}\right] < \mathbf{0} \ \textrm{,}
\end{equation}

\begin{equation}\label{sec3_eq23}
\left[\begin{matrix}
-\rho^{-1}\bm{I} & \bm{U}_{pi}\bm{Y}_p\\
* & -u_{pi,\max}^{2}\bm{X}_p^{}\\
\end{matrix}\right] < \mathbf{0} \ \textrm{,}
\end{equation}
\noindent then there exists an in-plane motion controller such that requirements (\ref{req1}), (\ref{req2}) and (\ref{req3}) are satisfied, and positive scalar $\rho$ is an upper bound of the quadratic cost function~\eqref{sec2_eq14}.
\end{theorem}

\begin{proof}[\textbf{Proof}]
Consider the Lyapunov function $V_p(t) = \bm{p}^T(t)\bm{P}_p\bm{p}(t)$, where $\bm{P}_p\in\mathbb{R}^{4\times4}$ is a positive symmetric matrix. Substituting~\eqref{sec3_eq20} into the derivative of $V_p(t)$, we have
\begin{equation}\label{sec3_eq24}
\dot{V}_p(t) = \textrm{sym}\left[\bm{p}^T(t)\bm{P}_p\left(\bm{A}_p + \upDelta\bm{A}_p - \bm{B}_p\bm{K}_p\right)\bm{p}(t)\right] \ \textrm{.}
\end{equation}
In order to optimize the cost function \eqref{sec2_eq14} and guarantee the asymptotic stability of in-plane motion, let inequalities \eqref{sec3_eq25} hold
\begin{equation}\label{sec3_eq25}
\dot{V}_p(t) < -\left[\bm{p}^T(t)\bm{Q}_p\bm{p}(t) + \bm{u}_p^T(t)\bm{R}_p\bm{u}_p(t)\right] < 0 \ \textrm{.}
\end{equation}
Integrating~\eqref{sec3_eq25} from 0 to $\infty$ and noticing that $\bm{p}(t) \rightarrow \mathbf{0}$ as $t \rightarrow \infty$, we get
\begin{equation}\label{sec3_eq26}
0 < J_p = \int_{0}^{\infty}\left[\bm{p}^T(t)\bm{Q}_p\bm{p}(t) + \bm{u}_p^T(t)\bm{R}_p^{}\bm{u}_p^{}(t)\right]{\rm d}t \leq V_p(0) \ \textrm{.}
\end{equation}
From~\eqref{sec3_eq26}, we know that when inequalities~\eqref{sec3_eq25} hold, $V_p(0) = \bm{p}^T(0)\bm{P}_p\bm{p}(0)$ will be an upper bound of the quadratic cost function $J_p$. Substituting~\eqref{sec2_eq11}, \eqref{sec3_eq19} and \eqref{sec3_eq24} into \eqref{sec3_eq25} yields
\begin{equation}\label{sec3_eq27}
\boldsymbol{\upPsi}_p + \bm{P}_p\bm{E}_{p1}\boldsymbol{\upLambda}_p\bm{E}_{p2} + \bm{E}_{p2}^T\boldsymbol{\upLambda}_p^T\left(\bm{P}_p\bm{E}_{p1}\right)^T < \mathbf{0} \ \textrm{,}
\end{equation}
\noindent where
\begin{flalign*}
\indent
\boldsymbol{\upPsi}_p = \textrm{sym}\left[\bm{P}_p\left(\bm{A}_p - \bm{B}_p\bm{K}_p\right)\right] + \bm{Q}_p + \bm{K}_p^T\bm{R}_p^{}\bm{K}_p^{} \ \textrm{.}&&
\end{flalign*}

\noindent Since $\boldsymbol{\upPsi}_p$ is a symmetric matrix, according to \hyperref[Lemma_LMI]{Lemma~2} and \eqref{sec2_eq11}, there exists a positive scalar $\varepsilon_p$ ensuring \eqref{sec3_eq27} by
\begin{equation}\label{sec3_eq28}
\boldsymbol{\upPsi}_p + \varepsilon_p\bm{P}_p\bm{E}_{p1}\left(\bm{P}_p\bm{E}_{p1}\right)^T + \varepsilon_p^{-1}\bm{E}_{p2}^T\bm{E}_{p2}^{} < \mathbf{0} \ \textrm{.}
\end{equation}

\noindent By Schur complement, inequality~\eqref{sec3_eq28} can be rewritten in a matrix form as
\begin{equation}\label{sec3_eq29}
\left[\begin{matrix}
\boldsymbol{\upPi}_{11} & \boldsymbol{\upPi}_{12}\\
* & \boldsymbol{\upPi}_{22}\\
\end{matrix}\right] < \mathbf{0} \ \textrm{,}
\end{equation}

\noindent where
\begin{flalign*}
\begin{split}
\indent
&\boldsymbol{\upPi}_{11} = \textrm{sym}\left[\bm{P}_p\left(\bm{A}_p - \bm{B}_p\bm{K}_p\right)\right] + \varepsilon_p\bm{P}_p^{}\bm{E}_{p1}^{}\bm{E}_{p1}^{T}\bm{P}_p^{T} \ \textrm{,}\\
\indent
&\boldsymbol{\upPi}_{12} = \left[\begin{matrix} \bm{E}_{p2}^T & \bm{K}_p^T & \bm{I}\end{matrix}\right] \ \textrm{,}\\
\indent
&\boldsymbol{\upPi}_{22} = \textrm{diag}\left(-\varepsilon\bm{I}, \ -\bm{R}_p^{-1}, \ -\bm{Q}_p^{-1} \right) \ \textrm{.}
\end{split}&
\end{flalign*}
\noindent With the variable substitutions, $\bm{X}_p = \bm{P}_p^{-1}$ and $\bm{Y}_p = \bm{K}_p\bm{P}_p^{-1}$, pre- and post-multiply~\eqref{sec3_eq29} with diag$(\bm{X}_p, \ \bm{I})$, and then \eqref{sec3_eq21} in \hyperref[theorem1]{Theorem~1} is obtained. To minimize $V_p(0)$, an upper bound of $J_p$, a positive scalar $\rho$ is introduced and meets
\begin{equation}\label{sec3_eq30}
V_p(0) = \bm{p}^T(0)\bm{P}_p\bm{p}(0) \leq \rho \ \textrm{.}
\end{equation}
\noindent By Schur complement, inequality~\eqref{sec3_eq30} is equivalent to 
\begin{equation}\label{sec3_eq31}
\left[\begin{matrix}
\rho & \bm{p}^T(0)\\
* & -\bm{P}_{p}^{-1}\\
\end{matrix}\right]
< \mathbf{0} \ \textrm{.}
\end{equation}

\noindent Pre- and post-multiplying~\eqref{sec3_eq31} with diag$(\rho^{-1}, \ \bm{I})$, the LMI constraint~\eqref{sec3_eq22} in \hyperref[theorem1]{Theorem~.1} is obtained. LMIs~\eqref{sec3_eq21} and \eqref{sec3_eq22} have fulfilled the requirements (\ref{req1}) and (\ref{req2}). In order to meet the requirement (\ref{req3}), squaring both sides of~\eqref{sec2_eq15} and dividing each side by $u_{pi,\max}^2$, then there is
\begin{equation}\label{sec3_eq32}
u_{pi, \max}^{-2}\left[\bm{U}_{pi}\bm{K}_p\bm{p}(t)\right]^T\bm{U}_{pi}\bm{K}_p\bm{p}(t) \leq 1 \ \textrm{.}
\end{equation}

\noindent Dividing both sides of \eqref{sec3_eq30} by $\rho$ and considering $\dot{V}_p(t) < 0$, we have
\begin{equation}\label{sec3_eq33}
\rho^{-1}V_p(t) < \rho^{-1}V_p(0) \leq 1 \ \textrm{.}
\end{equation}

\noindent Then we can guarantee the inequality~\eqref{sec3_eq32} by
\begin{equation}\label{sec3_eq34}
u_{pi, \max}^{-2}\left[\bm{U}_{pi}\bm{K}_p\right]^T\bm{U}_{pi}\bm{K}_p < \rho^{-1}\bm{P}_p \ \textrm{.}
\end{equation}

\noindent By Schur complement, inequality~\eqref{sec3_eq34} can be rewritten as
\begin{equation}\label{sec3_eq35}
\left[\begin{matrix}
-\rho^{-1}\bm{I} & \bm{U}_{pi}\bm{K}_p\\
* & -u_{pi,\max}^2\bm{P}_p\\
\end{matrix}\right] < \mathbf{0} \ \textrm{.}
\end{equation}
\noindent Pre- and post-multiplying~\eqref{sec3_eq35} with diag$(\bm{I}, \ \bm{X}_p)$, the LMI constraint~\eqref{sec3_eq23} in \hyperref[theorem1]{Theorem~1} is obtained. This completes the proof.
\end{proof}

It can be inferred from \eqref{sec3_eq30} that the quadratic cost function $J_p$ will be optimal if the positive scalar $\rho$ is minimized. Therefore, another positive scalar $\sigma$ is introduced and meets $\sigma > \rho$, which is equivalent to
\begin{equation}\label{sec3_eq36}
\left[\begin{matrix}
-\sigma & 1\\
1 & -\rho^{-1}\\
\end{matrix}\right]
< \mathbf{0} \ \textrm{.}
\end{equation}

\noindent \noindent Then combining \hyperref[theorem1]{Theorem~1} and \eqref{sec3_eq36}, the thrust-limited robust guaranteed cost controller for in-plane motion with initial state $\bm{p}(0)$ can be obtained by solving the following convex optimization problem
\begin{equation}\label{sec3_eq37}
\min_{\varepsilon_p, \ \rho^{-1}, \ \bm{X}_p, \ \bm{Y}_p} \sigma \ \textrm{,}
\end{equation}
\nolinenumbers
\begin{center}
s.t. \eqref{sec3_eq21}, \eqref{sec3_eq22}, \eqref{sec3_eq23} and \eqref{sec3_eq36}.
\end{center}

\noindent The state feedback gain matrix $\bm{K}_p$ can be solved by $\bm{K}_p^{} = \bm{Y}_p^{}\bm{X}_p^{-1}$. 
\begin{remark}
Since no preassigned parameter is needed in \hyperref[theorem1]{Theorem~1}, the motion controller obtained from \eqref{sec3_eq37} is less conservative and therefore more practical than the controllers employed in \cite{H.Gao_TCST_2009} and \cite{X.Yang_AST_2013} when implemented to spacecraft rendezvous, which can be drawn from the minimum feasible upper bounds of actuators.
\end{remark}

\subsection{Out-of-Plane Motion Controller}
Consider the following state feedback control law
\begin{equation}\label{sec3_eq38}
u_q(t) = -\bm{K}_q\bm{q}(t) \ \textrm{,}
\end{equation}

\noindent where $\bm{K}_q \in \mathbb{R}^{1\times2}$ is the state feedback gain matrix of the out-of-plane motion controller. Substituting~\eqref{sec3_eq38} into the plant model~\eqref{sec2_eq12}, the closed-loop model for out-of-plane motion is
\begin{equation}\label{sec3_eq39}
\dot{\bm{q}}(t) = \left(\bm{A}_q + \upDelta\bm{A}_q - \bm{B}_q\bm{K}_q\right)\bm{q}(t) + \bm{B}_qw_q(t) \ \textrm{.}
\end{equation}

\noindent To optimize cost function $J_q$ in the presence of external disturbance $w_q(t)$, define a controlled output as
\begin{equation}\label{sec3_eq40}
z_q(t) = \bm{Q}_q^{\frac{1}{2}}\bm{q}(t) + R_q^{\frac{1}{2}}u_q(t) \ \textrm{.}
\end{equation}

\noindent Then requirement (\ref{req5}) can be fulfilled by minimizing $\| z_q(t) \|_2$, which is assumed to be bounded by
\begin{equation}\label{sec3_eq41}
\| z_q(t) \|_2 \leq \gamma \| w(t) \|_2 \ \textrm{,}
\end{equation}

\noindent where $\gamma$ is the $H_\infty$ performance. Sufficient condition for the existence of a robust $H_\infty$ controller is given in \hyperref[theorem2]{Theorem~2}.

\begin{theorem}\label{theorem2}
Consider the closed-loop system~\eqref{sec3_eq39} with the state feedback control law in~\eqref{sec3_eq38}. If there exist a positive symmetric matrix $\bm{X}_q \in \mathbb{R}^{2\times2}$, a matrix $\bm{Y}_q\in\mathbb{R}^{1\times2}$ and a positive scalar $\varepsilon_q$ satisfying
\begin{equation}\label{sec3_eq42}
\left[\begin{matrix}
\textrm{sym}\left(\bm{A}_q\bm{X}_q - \bm{B}_q\bm{Y}_q\right) + \varepsilon_q^{}\bm{E}_{q1}^{}\bm{E}_{q1}^T & \bm{B}_q & \bm{X}_q^{}\bm{E}_{q2}^T & \mathbf{0} & \bm{Y}_q^T & \bm{X}_q\\
* & -\gamma^2\bm{I} & \mathbf{0} & \mathbf{0} & \mathbf{0} & \mathbf{0}\\
* & * & -\varepsilon_q\bm{I} & \mathbf{0} & \mathbf{0} & \mathbf{0}\\
* & * & * & -\varepsilon_q\bm{I} & \mathbf{0} & \mathbf{0}\\
* & * & * & * & -R_q^{-1}\bm{I} & \mathbf{0}\\
* & * & * & * & * & -\bm{Q}_q^{-1}\\
\end{matrix}\right] < \mathbf{0} \  \textrm{,}
\end{equation}
\noindent then there exists an in-plane motion controller such that requirements (\ref{req4}) and (\ref{req5}) are satisfied.
\end{theorem}

\begin{proof}[\textbf{Proof}]
Consider the Lyapunov function $V_q(t) = \bm{q}^T(t)\bm{P}_q\bm{q}(t)$, where $\bm{P}_q \in \mathbb{R}^{2\times2}$ is a positive symmetric matrix. Substituting~\eqref{sec3_eq39} into the derivative of $V_q(t)$, there is
\begin{equation}\label{sec3_eq43}
\dot{V}_q(t) = \left[\begin{matrix} \bm{q}(t) \\ w_q(t) \end{matrix}\right]^T
\left[\begin{matrix}
\textrm{sym}\left[\bm{P}_p\left(\bm{A}_q + \upDelta\bm{A}_q - \bm{B}_q\bm{K}_q  \right)\right] & \bm{P}_q\bm{B}_q\\
* & 0 \\
\end{matrix}\right]
\left[\begin{matrix} \bm{q}(t) \\ w_q(t) \end{matrix} \right] \ \textrm{.}
\end{equation}

\noindent Assuming external disturbance $w_q(t)$ to be 0, the derivate of $V_q(t)$ becomes
\begin{equation}\label{sec3_eq44}
\dot{V}_{q0}(t) = \textrm{sym}\left[\bm{q}^T(t)\bm{P}_q\left(\bm{A}_q + \upDelta\bm{A}_q - \bm{B}_q\bm{K}_q\right)\bm{q}(t)\right] \ \textrm{.}
\end{equation}

\noindent Squaring both sides of~\eqref{sec3_eq41}, there is
\begin{equation}\label{sec3_eq45}
z_q^T(t)z_q(t) - \gamma^2w_q^T(t)w_q(t) \leq 0 \ \textrm{.}
\end{equation}

\noindent Integrating~\eqref{sec3_eq45} from 0 to $\infty$, we have
\begin{equation}\label{sec3_eq46}
\int_{0}^{\infty}\left[z_q^T(t)z_q(t) - \gamma^2w_q^T(t)w_q(t) + \dot{V}_q(t)\right]dt + V_q(0) - V_q(\infty) \leq 0 \ \textrm{.}
\end{equation}

\noindent According to \hyperref[asm2]{Assumption~2} of zero-initial condition and the fact $V_q(\infty) > 0$, inequalities~\eqref{sec3_eq41}, \eqref{sec3_eq45} and \eqref{sec3_eq46} can be guaranteed by
\begin{equation}\label{sec3_eq47}
z_q^T(t)z_q(t) - \gamma^2w_q^T(t)w_q(t) + \dot{V}_q(t) \leq 0 \ \textrm{.}
\end{equation}

\noindent  Substituting \eqref{sec3_eq40} and \eqref{sec3_eq43} into \eqref{sec3_eq47}, we can obtain
\begin{equation}\label{sec3_eq48}
\left[\begin{matrix}
\textrm{sym}\left[\bm{P}_p\left(\bm{A}_q + \upDelta\bm{A}_q - \bm{B}_q\bm{K}_q\right)\right] + \bm{Q}_q + \bm{K}_q^TR_q\bm{K}_q & \bm{P}_q\bm{B}_q\\
* & -\gamma^2\bm{I}\\
\end{matrix}\right] < \mathbf{0} \ \textrm{.}
\end{equation}

\noindent By Schur complement, inequality~\eqref{sec3_eq48} can be rewritten as
\begin{equation}\label{sec3_eq49}
\boldsymbol{\upTheta}_1 < \boldsymbol{\upTheta}_2 \ \textrm{,}
\end{equation}

\noindent where
\begin{flalign*}
\begin{split}
\indent&
\boldsymbol{\upTheta}_1 = \textrm{sym}\left[\bm{P}_q\left(\bm{A}_q + \upDelta\bm{A}_q - \bm{B}_q\bm{K}_q\right)\right] \ \textrm{,}\\
\indent&
\boldsymbol{\upTheta}_2 = -\bm{Q}_q - \bm{K}_q^TR_q\bm{K}_q - \gamma^{-2}\bm{P}_q\bm{B}_q\left(\bm{P}_q\bm{B}_q\right)^T \ \textrm{.}
\end{split}&
\end{flalign*}

\noindent From \eqref{sec3_eq49}, we can learn that $\boldsymbol{\upTheta}_2 < \mathbf{0}$; thus $\boldsymbol{\upTheta}_1 < \mathbf{0}$ and $\dot{V}_{q0} < 0$, {\it i.e.}, inequality \eqref{sec3_eq48} guarantees the stabilities of the nominal model (without disturbance) as well as the perturbed model \eqref{sec2_eq12}, which fulfills requirement (\ref{req4}). Substituting ~\eqref{sec2_eq13} into ~\eqref{sec3_eq48}, we have
\begin{equation}\label{sec3_eq50}
\boldsymbol{\upPsi}_q + \boldsymbol{\upDelta}_q\boldsymbol{\upPhi}_q\mathbf{E}_q + \mathbf{E}_q^T\boldsymbol{\upPhi}_q^T\boldsymbol{\upDelta}_q^T  < \mathbf{0} \ \textrm{,}
\end{equation}

\noindent where
\begin{flalign*}
\begin{split}
\indent&
\boldsymbol{\upPsi}_q = \left[\begin{matrix}
\textrm{sym}\left[\bm{P}_q\left(\bm{A}_q - \bm{B}_q\bm{K}_q\right)\right] + \bm{Q}_q + \bm{K}_q^TR_q\bm{K}_q &  \bm{P}_q\bm{B}_q\\
* & -\gamma^2\bm{I}\\
\end{matrix}\right] \ \textrm{,}\\
\indent&
\boldsymbol{\upDelta}_q = 
\left[\begin{matrix}
\bm{P}_q\bm{E}_{q1} & \mathbf{0}\\
\mathbf{0} &  \mathbf{0}\\
\end{matrix}\right] \ \textrm{,}
\qquad
\boldsymbol{\upPhi}_q = \left[\begin{matrix}
\boldsymbol{\upLambda}_q & \mathbf{0}\\
\mathbf{0} & \mathbf{0}\\
\end{matrix}\right] \ \textrm{,}
\qquad
\mathbf{E}_q = \left[\begin{matrix}
\bm{E}_{q2} & \mathbf{0}\\
\mathbf{0} & \mathbf{0}\\
\end{matrix}\right] \ \textrm{.}
\end{split}&
\end{flalign*}

\noindent Since $\boldsymbol{\upPsi}_q$ is a symmetric matrix, according to \hyperref[Lemma_LMI]{Lemma~2} and \eqref{sec2_eq13}, there exists a positive scalar $\varepsilon_q$ ensuring~\eqref{sec3_eq50} by
\begin{equation}\label{sec3_eq51}
\boldsymbol{\upPsi}_q + \varepsilon_q\boldsymbol{\upDelta}_q^{}\boldsymbol{\upDelta}_q^{T} + \varepsilon_q^{-1}\mathbf{E}_q^{T}\mathbf{E}_q < \mathbf{0} \ \textrm{.}
\end{equation}

\noindent By Schur complement, inequality~\eqref{sec3_eq51} is equivalent to
\begin{equation}\label{sec3_eq52}
\left[\begin{matrix}
\boldsymbol{\upOmega}_{11} & \boldsymbol{\upOmega}_{12}\\
* & \boldsymbol{\upOmega}_{22}\\
\end{matrix}\right] < \mathbf{0} \ \textrm{,}
\end{equation}

\noindent where
\begin{flalign*}
\begin{split}
\indent&
\boldsymbol{\upOmega}_{11} = 
\left[\begin{matrix}
\textrm{sym}\left[\bm{P}_q\left(\bm{A}_q - \bm{B}_q\bm{K}_q\right)\right] + \varepsilon_q^{}\bm{P}_q^{}\bm{E}_{q1}^{}\bm{E}_{q1}^T\bm{P}_q^T & \bm{P}_q\bm{B}_q \\
* & -\gamma^2\bm{I}\\
\end{matrix}\right] \textrm{,}
\qquad
\boldsymbol{\upOmega}_{12} = 
\left[\begin{matrix}
\bm{E}_{q2}^T & \mathbf{0} & \bm{K}_q^T & \bm{I}\\
\mathbf{0} & \mathbf{0} & \mathbf{0} & \mathbf{0}\\
\end{matrix}\right] \textrm{,}\\
\indent&
\boldsymbol{\upOmega}_{22} = \textrm{diag}\left(-\varepsilon\bm{I}, \ \varepsilon\bm{I}, \ -R_q^{-1}, -\bm{Q}_q^{-1}  \right) \textrm{.}
\end{split}&
\end{flalign*}

\noindent Define the variable substitutions $\bm{X}_q = \bm{P}_q^{-1}$ and $\bm{Y}_q = \bm{K}_q\bm{P}_q^{-1}$. Pre- and post-multiplying~\eqref{sec3_eq52} with diag$(\bm{X}_q, \bm{I})$, the LMI constraint~\eqref{sec3_eq42} is obtained. This completes the proof.
\end{proof}

\begin{remark}
In the proof of \hyperref[{theorem2}]{Theorem~2}, zero-initial condition has been utilized to synthesize the robust $H_\infty$ controller, which is a reasonable simplification for the engineering problem described in this paper. However, for the situation when zero-initial condition is not satisfied, some extended robust $H_\infty$ control methods can be adopted, which have been discussed by \cite{P.P.Khargonekar_SIAMJCO_1991}, \citet{T.Namerikawa_ACC_2002}, \cite{A.V.Savkin_TAES_2003} and \cite{Y.K.Foo_TCS_2006}. Nevertheless, due to the implicit expressions of $H_\infty$ performance and more rigorous assumptions, extended robust $H_\infty$ controllers are not frequently employed in the existing literatures.
\end{remark}

The robust $H_\infty$ controller for out-of-plane motion can be obtained by solving the following convex optimization problem 
\begin{equation}\label{sec3_eq53}
\min_{\varepsilon_q, \ \bm{X}_q, \ \bm{Y}_q}\gamma \ \textrm{,}
\end{equation}
\nolinenumbers
\begin{center}
s.t. \eqref{sec3_eq42} .
\end{center}

\noindent State feedback gain matrix $\bm{K}_q$ can be determined by  $\bm{K}_q = \bm{Y}_q^{}\bm{X}_q^{-1}$. With the state feedback gain matrices $\bm{K}_p$ and $\bm{K}_q$ solved from \eqref{sec3_eq37} and \eqref{sec3_eq53}, a partially independent control scheme for spacecraft rendezvous can be constructed, and we will discuss this procedure detailedly in the next section with an illustrative example.

\section{Illustrative Example}\label{sec4}
In this section, a comparison between the partially independent and coupled control schemes will be conducted to illustrate the advantages of the former. All the simulation results were obtained from a two-body model:

\begin{subequations}\label{sec4_eq1}
\nolinenumbers\begin{equation}
\ddot{\bm{r}} + \frac{\mu}{r^3}\bm{r} = \bf{0}\textrm{,}
\end{equation}
\begin{equation}
\ddot{\bm{r}}_c + \frac{\mu}{r^3_c}\bm{r}_c = \frac{\bm{u}+\bm{w}}{m}\textrm{,}
\end{equation}
\end{subequations}

\noindent where the position vectors $\bm{r}$ and $\bm{r}_c$ have been defined in \hyperref[fig1]{Figure 1} and satisfy $\bm{r}_c - \bm{r} = \bm{x}(t)$; $m$ and $\bm{u}$ are the mass and control vector of the chase vehicle; and $\bm{w}$ is the disturbance, which consists of long and short period perturbations along the $z-$axis.

Consider a rendezvous scenario as follows. A target vehicle is in a low earth orbit (LEO) with eccentricity $e = 0.05$ and semimajor axis $a = 7082.253$ km; then we can figure out that the mean motion of the target vehicle is $n = 1.059\times10^{-3}$ rad/s, {\it i.e.}, the period of the reference orbit is $T = 5931.53$ s; the initial state vector is $x(0) = \left[-5000, 5000, 0, 5, -5, 0\right]$, and the mass of the chase vehicle is $m = 500$ kg. When solving the convex problem \eqref{sec3_eq37}, the minimum feasible upper bound of the in-plane propulsion is $6.8$ N; nevertheless, considering the discrepancy between the plant models (\ref{sec2_eq10}, \ref{sec2_eq12}) and simulation model (\hyperref[sec4_eq1]{54a-b}), the upper bounds of the in-plane propulsions are set $u_{px,\max} = u_{py,\max} = 15$ N to guarantee the robustness of the controllers. In \eqref{sec2_eq12}, consider an extreme unknown out-of-plane disturbance that may rarely exist in reality
\begin{equation}\label{sec4_eq54}
w_q(t) = 4.3\sin\left(1.059\times10^{-3}t\right) + 0.5\sin\left(0.1059t\right) \textrm{.}
\end{equation}

\noindent where the first term represents long period perturbation caused by the nonhomogeneity of central planet and gravitational forces from other celestial bodies, {\it etc}; while the second term represents short period perturbation caused by the solar wind, atmospheric drag, {\it etc}. Therefore, in this example, the upper bound of the out-of-plane propulsion is set $\bm{u}_{q,\max} = 5$~N, which is greater than the maximum disturbance $w_{q,\max} \approx 4.8$~N. All the weighting matrices and scalar, $\bm{Q}_p, \bm{Q}_q, \bm{R}_p \textrm{ and } R_q $, are assigned to be units. With these parameters, a partially independent controller and a coupled controller for comparison are to be solved in the following sections.

\subsection{Partially Independent Controller}
The partially independent control scheme can be synthesized by solving \eqref{sec3_eq37} and \eqref{sec3_eq53}. For in-plane motion controller \eqref{sec3_eq37}, the initial state vector is $\bm{p}(0) = \left[-5000, \ 5000, \ 5, \ -5\right]^T$, and the matrices $\bm{E}_{p1}$, $\bm{E}_{p2}$ and $\boldsymbol{\upLambda}_{p}$ in \eqref{sec2_eq11} are assigned as follows:

\begin{align}\label{sec4_eq55}
\bm{E}_{p1} = \left[\begin{matrix}
0 & 0 & 0 & 0\\
0 & 0 & 0 & 0\\
0 & 2e & 4e & 0\\
2e & 0 & 0 & 4e\\
\end{matrix}\right] \textrm{,}
\qquad
\bm{E}_{p2} = \left[\begin{matrix}
n^2 & 0 & 0 & 0\\
0 & n^2 & 0 & 0\\
2.5n^2 & 0 & 0 & n\\
0 & 0.25n^2 & -n & 0\\
\end{matrix}\right] \textrm{,}
\end{align}

\vspace{-1.5em}
\begin{align*}
\boldsymbol{\upLambda}_{p} = \textrm{diag}\left(\sin M, \ -\sin M, \ \cos M, \ \cos M \right) \textrm{,}
\end{align*}

\noindent where the mean anomaly $M = nt$. Then solving \eqref{sec3_eq37}, the state feedback gain matrix for in-plane motion controller is obtained
\begin{equation}\label{sec4_eq56}
\bm{K}_p = \left[\begin{matrix}
\bm{K}_{p,11} & \bm{K}_{p,12}\\
\end{matrix} \right] 
= \left[\begin{matrix}
0.0024 & -0.0013 & 0.7535 & 0.0593\\
0.0015 & 0.0010 & 0.2952 & 1.3332\\
\end{matrix}\right]\textrm{.}
\end{equation}
\noindent where $\bm{K}_{p,11} \textrm{ and } \bm{K}_{p,12} \in \mathbb{R}^{2\times2}$. For the out-of-plane motion controller~\eqref{sec3_eq53}, the initial state vector is $\bm{q}(0) = [0, \ 0]^T$, and the matrices $\bm{E}_{q1}$, $\bm{E}_{q2}$ and $\boldsymbol{\upLambda}_{q}$ in \eqref{sec2_eq13} are assigned as follows:
\begin{equation}\label{sec4_eq57}
\bm{E}_{q1} = \left[\begin{matrix}
0 & 0\\
6e & 0\\
\end{matrix}\right]\textrm{,}
\qquad
\bm{E}_{q2} = \left[\begin{matrix}
n^2 & 0\\
0 & 0\\
\end{matrix}\right]\textrm{,}
\qquad
\boldsymbol{\upLambda}_{q} = \left[\begin{matrix}
-0.5\cos M & 0\\
0 & 0\\
\end{matrix}\right]\textrm{.}
\end{equation}

\noindent Solving \eqref{sec3_eq53}, the optimal $H_\infty$ performance is $\gamma = 1.000778383$, and the state feedback gain matrix for out-of-plane motion controller is
\begin{equation}\label{sec4_eq58}
\bm{K}_q  
= \left[\begin{matrix} K_{q,11} & K_{q,12} \end{matrix}\right]
= \left[\begin{matrix} 196.8030 & 5.8353\times10^4 \end{matrix}\right] \textrm{.}
\end{equation}

\noindent Combing \eqref{sec4_eq56} and \eqref{sec4_eq58} together, the state feedback gain matrix for partially independent controller is
\begin{equation}\label{sec4_eq59}
\bm{K}_{pic} = \left[\begin{matrix} \bm{K}_{p,11} & \mathbf{0}_{2\times1} & \bm{K}_{p,12} & \mathbf{0}_{2\times1}\\
\mathbf{0}_{1\times2} & K_{q,11} & \mathbf{0}_{1\times2} & K_{q,12}\\
\end{matrix}\right] \textrm{.}
\end{equation}

\noindent where $\bm{K}_{pic} \in \mathbb{R}^{3\times6}$, and the control vector in (\hyperref[sec4_eq1]{54b}) is generated by $\bm{u}_{pic}(t) = -\bm{K}_{pic}\bm{x}(t)$.

\subsection{Coupled Controller}
In \hyperref[sec3]{Section 3}, the in-plane motion controllers for $x-$ and $y-$axis were synthesized jointly, while the out-of-plane motion controller was designed independently. To verify the advantages of this scheme in robustness, we will introduce a coupled rendezvous controller in this section for comparison. The coupled control scheme synthesizes $x-$, $y-$ and $z-$axis controllers together and meets the requirements similar as \eqref{req1}, \eqref{req2} and \eqref{req3}; therefore, the coupled controller can be attained by solving a convex optimization problem similar as \hyperref[theorem1]{Theorem~1}. For brevity, the result of coupled control scheme will be given directly, while the detailed derivations of it will not be included in this paper. However, some similar procedures for synthesizing a coupled controller can be found in \cite{X.Yang_AST_2013}, \cite{D.Sheng_MPE_2014} and \cite{N.Wan_MPE_2013, N.Wan_AAA_2014}. With the same parameters assigned  in previous sections, the control vector in (\hyperref[sec4_eq1]{54b}) for coupled control scheme is generated by $\bm{u}_{cc}(t) = -\bm{K}_{cc}\bm{x}(t)$, where the state feedback gain matrix $\bm{K}_{cc}$ is
\begin{equation}\label{sec4_eq60}
\bm{K}_{cc} = \left[\begin{matrix}
0.0024 & -0.0014 & 2.1542\times10^{-4} & 0.8445 & 0.0467 & 0.1198\\
0.0017 & 7.487\times10^{-4} & -4.3822\times10^{-4} & 0.5689 & 1.3525 & 0.1901\\
3.5306\times10^{-4} & -2.0446\times10^{-4} & 5.2548\times10^{-4} & 0.1792 & 0.0234 & 0.7065\\
\end{matrix}\right] \textrm{,}
\end{equation}

\subsection{Simulation Results}
All the simulation data are collected from the two-body model (\hyperref[sec4_eq1]{4.1a-b}), which is more adjacent to the practical circumstance than plant models \eqref {sec2_eq9}, \eqref {sec2_eq10} and \eqref {sec2_eq12}. The simulation results of in-plane and out-of-plane motions will be shown successively as follows.

\subsubsection*{In-Plane Motion}
The relative in-plane trajectories of the chase vehicles with different control schemes are depicted in \hyperref[fig2]{Figure 2}. The in-plane distances and control propulsions of the chase vehicle with partially independent control scheme are illustrated in \hyperref[fig3]{Figure 3} and \hyperref[fig4]{Figure 4} respectively.

\setlength{\abovecaptionskip}{-3pt} 
\setlength{\belowcaptionskip}{-10pt}
\begin{figure}[htpb]
\centering
\includegraphics[width=0.5\textwidth]{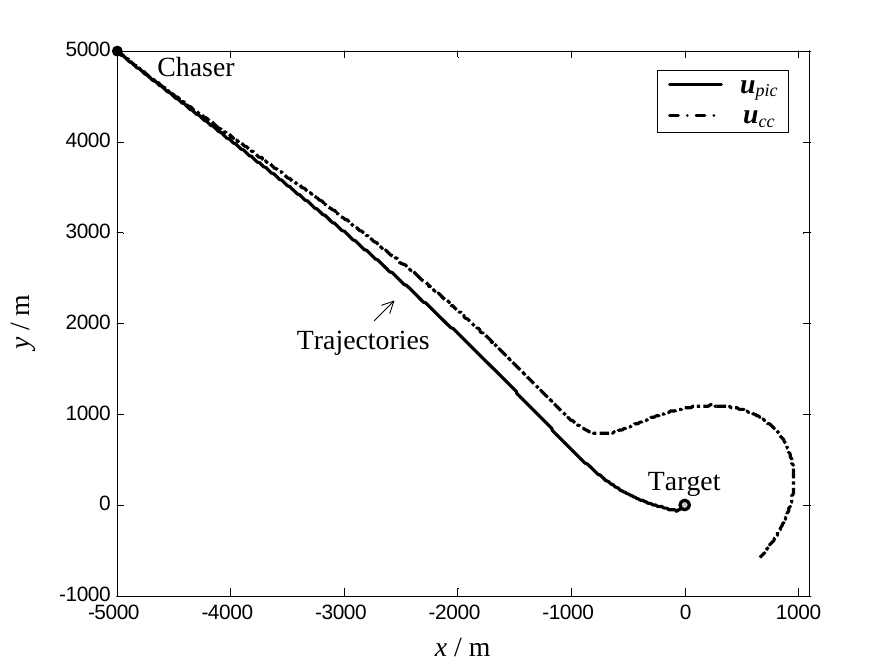}
\caption{In-plane rendezvous trajectories in first 5000 s.}\label{fig2}
\end{figure}
\begin{figure}[htpb]
\centering
\includegraphics[width=0.5\textwidth]{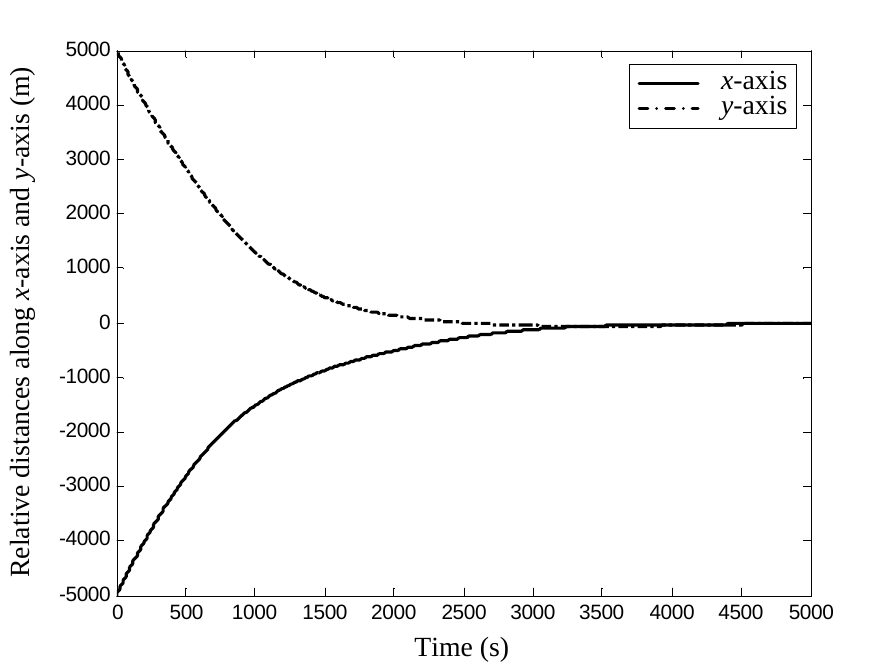}
\caption{In-plane relative distances between two spacecraft in the first 5000 s.}\label{fig3}
\end{figure}
\begin{figure}[htpb]
\centering
\includegraphics[width=0.5\textwidth]{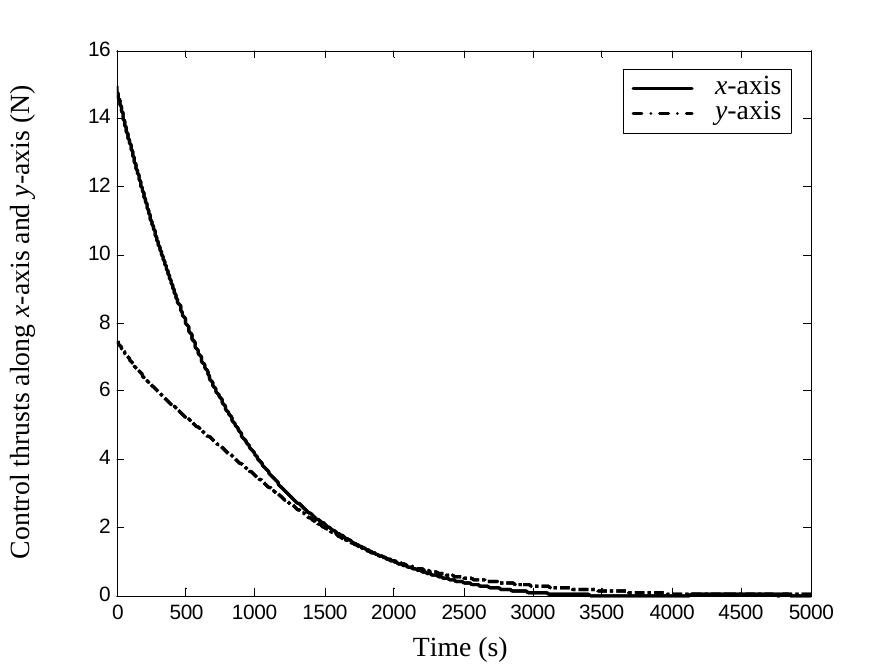}
\caption{In-plane control propulsions of chase vehicle in the first 5000 s.}\label{fig4}
\end{figure}

\begin{remark}
\hyperref[fig2]{Figure 2} and \hyperref[fig3]{Figure 3} show that the partially independent controller $\bm{u}_{pic}(t)$ fulfilled the requirement \eqref{req1}, asymptotic stability at $\bm{p}(t) = \mathbf{0}$, while the coupled controller $\bm{u}_{cc}(t)$ failed in finishing the rendezvous, which is one of the advantages of $\bm{u}_{pic}(t)$ over $\bm{u}_{cc}(t)$. \hyperref[fig4]{Figure 4} shows that the in-plane control propulsions of the chase vehicle with $\bm{u}_{pic}(t)$ are restricted below the upper bounds $u_{px,\max} = u_{py,\max} = 15$~N, which fulfilled the requirement \eqref{req3}.
\end{remark}

\subsubsection*{Out-of-Plane Motion}
The out-of-plane distances and control propulsions of the chase vehicles with different control schemes are illustrated in \hyperref[fig5]{Figure 5} and \hyperref[fig6]{Figure 6}. \hyperref[fig7]{Figure~7} depicts the overall performance costs of the rendezvouses with different schemes. 

\begin{figure}[htpb]
\centering
\includegraphics[width=1\textwidth]{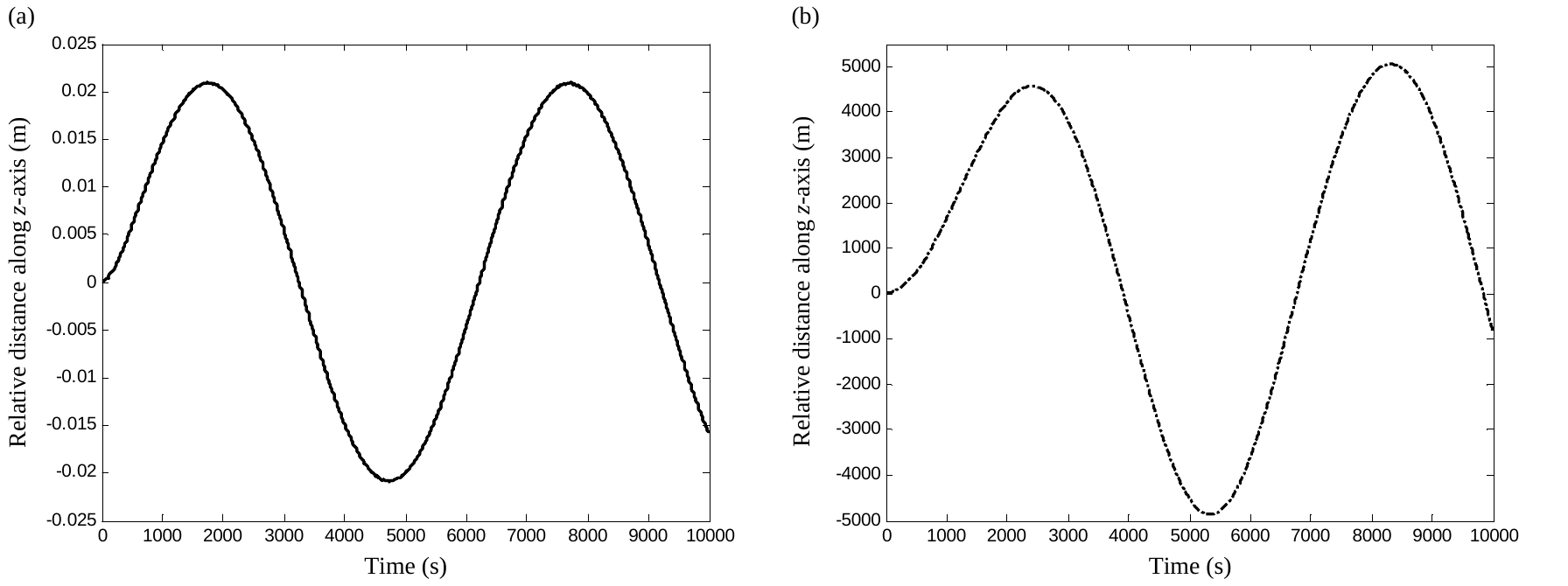}
\caption{Relative out-of-plane distance between two spacecrafts in the first 10000 s. (a) Out-of-plane distance with partially independent controller $\bm{u}_{pic}(t)$. (b) Out-of-plane distance with coupled controller $\bm{u}_{cc}(t)$.}\label{fig5}
\end{figure}
\begin{figure}[h]
\centering
\includegraphics[width=1\textwidth]{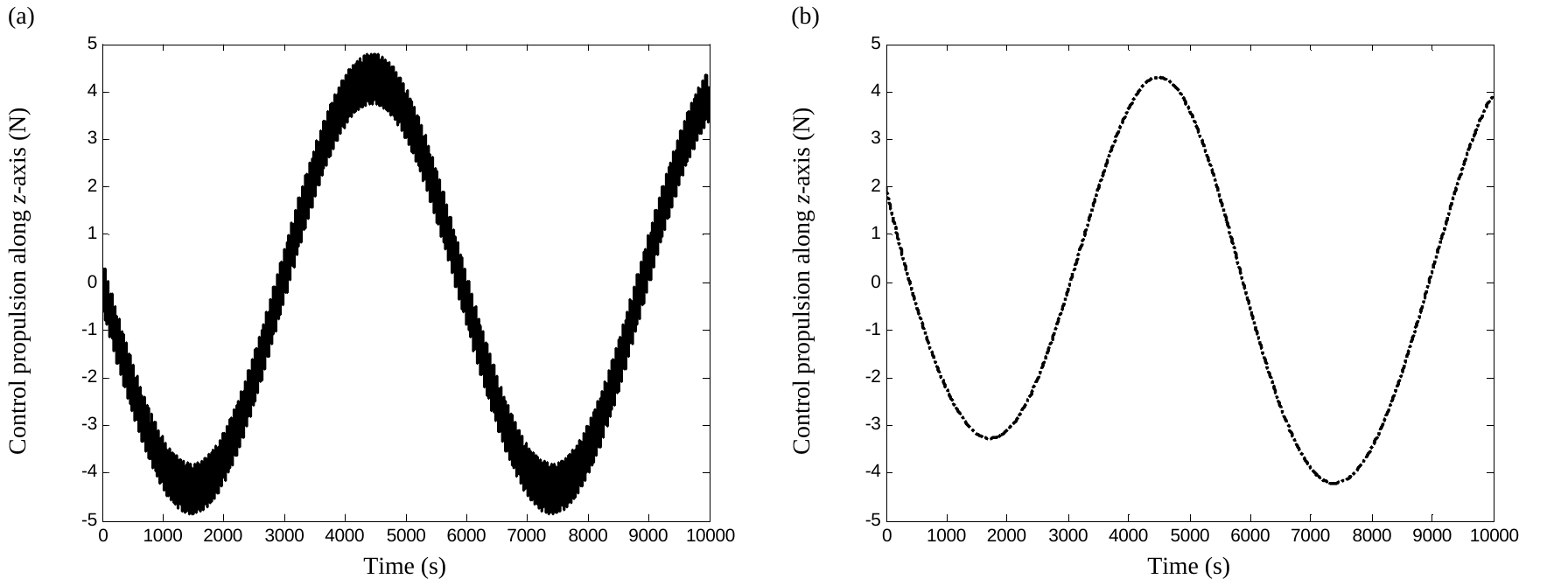}
\caption{Out-of-plane control propulsion of chase vehicle in the first 10000 s. (a) Out-of-plane control propulsion with partially independent controller $\bm{u}_{pic}(t)$. (b) Out-of-plane control propulsion with coupled controller $\bm{u}_{cc}(t)$.}\label{fig6}
\end{figure}
\begin{figure}[h]
\centering
\includegraphics[width=0.5\textwidth]{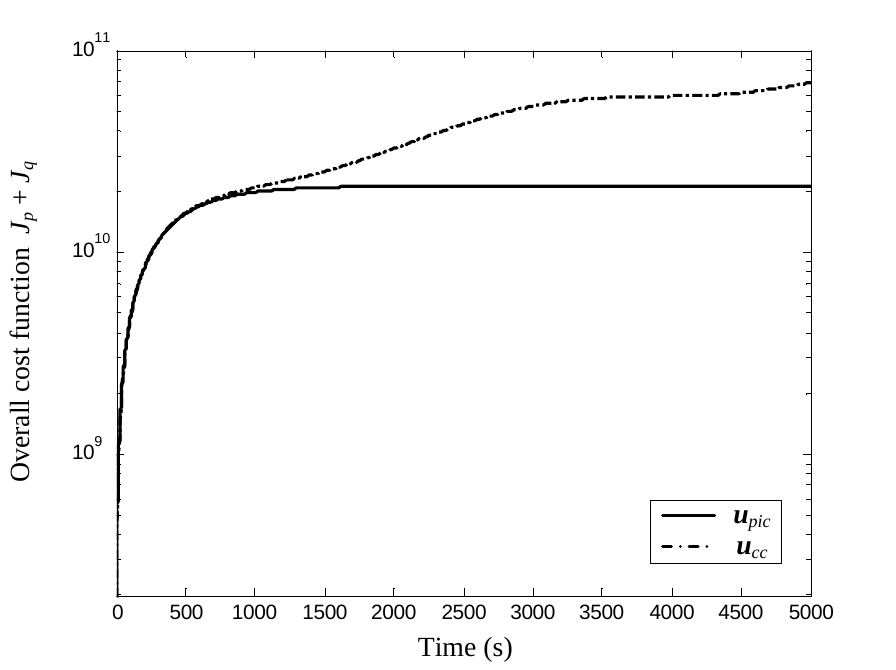}
\caption{Overall cost function in the first 5000 s.}\label{fig7}
\end{figure}

\begin{remark}
From \hyperref[fig5]{Figure 5}, we can conclude that the partially independent controller $\bm{u}_{pic}(t)$ fulfilled the requirement \eqref{req4}, robust stability at $\bm{q}(t) = \mathbf{0}$, while the coupled controller $\bm{u}_{cc}(t)$ failed again. From \hyperref[fig6]{Figure 6}, we can find that $\bm{u}_{pic}(t)$ tracked and suppressed the disturbance $\bm{w}(t)$ well, while the disturbance rejection ability of $\bm{u}_{cc}(t)$ was very poor; moreover, the magnitude of the out-of-plane propulsion was bounded and proportional to the magnitude of $w_q(t)$, which made our control method practical for engineering applications. From \hyperref[fig7]{Figure 7}, we can find that although the coupled control scheme optimize the overall cost function jointly, when out-of-plane disturbance exists, the overall cost function of the partially independent control scheme is much lower, which is another advantage of $\bm{u}_{pic}(t)$ over $\bm{u}_{cc}(t)$.
\end{remark}

\section{Conclusions}\label{sec5}
In sum, this paper has proposed a partially independent control scheme for thrust-limited rendezvous in near-circular orbits. Based on the two-body problem, a linearized dynamical model for near-circular rendezvous has been established. An anti-windup robust guaranteed controller for in-plane motion and a robust $H_\infty$ controller have been synthesized to construct the partially independent control scheme. Finally, a comparative simulation has been employed to verify the advantages of the partially independent scheme over the coupled one. Due to its robust stability, optimal performance cost and bounded control propulsion, the partially independent control scheme has a wide range of application in spacecraft rendezvous.

\section{Acknowledgment}\label{sec6}
The authors specially acknowledge the staff and the readers from arXiv.org.

\bibliography{reference}
\addtolength{\itemsep}{-0.5em}

\end{document}